\documentclass[hidelinks,onefignum,onetabnum]{siamart251216}
\usepackage{mathtools}     
\usepackage{mathrsfs}     
\usepackage{bm}            
\usepackage{enumitem}      
\usepackage{diagbox}  
\usepackage{amsmath}
\usepackage{amssymb}
\usepackage{subcaption}   
\usepackage{float}          
\usepackage{algpseudocode}
\usepackage{listings} 
\newcommand{\ds}{\displaystyle}
\theoremstyle{plain}

\begin{document}
\title{Bernoulli--Strang--Fix Conditions:
Approximation and Prediction by Sampling Kantorovich Operators}

\headers{Sreya T and A. Antony Selvan}
{BSF Conditions: Approximation and Prediction}

\author{
Sreya T\thanks{Indian Institute of Technology (Indian School of Mines) Dhanbad, Dhanbad 826004, India
(\email{sreyasoman8862@gmail.com}).}
\and
A. Antony Selvan\thanks{
Indian Institute of Technology (Indian School of Mines) Dhanbad, Dhanbad 826004, India
(\email{antonyaans@gmail.com}, corresponding author).}
}

\maketitle

\begin{abstract}
In this paper, we introduce the Bernoulli–Strang–Fix conditions and their generalized versions for a vector-valued generator $\varphi=(\varphi_0,\dots, \varphi_{\rho-1})$ and a periodic nonuniform sampling set $X$. We use these conditions to establish exact and asymptotic polynomial reproduction properties of sampling Kantorovich operators associated with $(\varphi, X)$ up to a prescribed degree. We analyze the approximation properties and convergence behavior of these operators in detail. Furthermore, we demonstrate their application to signal prediction from a finite number of past local average samples, showing that sampling Kantorovich operators can also serve as effective prediction operators. Finally, we present numerical examples based on Gaussian functions and B-splines to illustrate and validate the theoretical approximation and prediction results.
\end{abstract}



\section{Introduction}
Sampling theory is a fundamental area of research with wide-ranging applications in signal processing, approximation theory, and harmonic analysis. The classical Shannon sampling theorem establishes an exact reconstruction formula for band-limited signals using their uniformly spaced pointwise samples. Over the years, this foundational result, based on pointwise samples, whether uniform or nonuniform, has been significantly extended to broader classes of function spaces, including shift-invariant spaces, spaces modeling signals with finite rate of innovation, Sobolev spaces, and wavelet spaces. However, in many practical situations, exact pointwise measurements are not accessible. Physical measurement devices are often unable to capture instantaneous values of a signal. Instead, the acquired data represent local averages of the signal over small intervals or neighborhoods. Such averaged measurements arise naturally in applications including image acquisition, communication systems, biomedical signal processing, and sensor networks. Therefore, reconstruction methods based solely on pointwise samples may not accurately model real-world sampling processes.

Average sampling theory provides a natural framework for studying signal reconstruction from such local integral measurements rather than exact pointwise values. This approach yields reconstruction formulas that are both mathematically rigorous and well suited to practical acquisition systems. Sampling Kantorovich operators are closely connected to this approach, as they replace exact sample values with averages of a function over suitable intervals around them.
The sampling Kantorovich operator, introduced by Butzer \textit{et al.} \cite{Bardarobutzer}, is defined by
\begin{equation}\label{K1butzer}
(K_W^{\varphi}f)(t)=
\sum_{l\in\mathbb Z}
W\left[
\int_{\frac{l}{W}}^{\frac{l+1}{W}}f(y)\,dy
\right]
\varphi\left(Wt-l\right), \quad t \in \mathbb{R},
\end{equation}
where $W>0$, $f$ is a locally integrable function on $\mathbb{R}$, and $\varphi$ is a generator satisfying appropriate assumptions. Practically, the sampling Kantorovich operators reduce time-jitter errors by calculating information from a neighborhood of each sampling point rather than exactly at that point. Several studies have widely investigated their approximation and convergence properties in various settings, including $L^p$ spaces and more general Orlicz spaces; see \cite{Bardarobutzer, Butzer1, Butzer2, Orlova, Vinti}. These methods have also successfully supported applications in image reconstruction and enhancement; see \cite{Costarelli1, Costarelli2}.

An important aspect in the study of approximation operators is their ability to reproduce polynomials. The classical Strang--Fix conditions characterize polynomial reproduction of quasi-interpolation operators in terms of the Fourier transform of the generator and play a fundamental role in approximation theory and finite element analysis. However, polynomial reproduction is not necessary for achieving high-order approximation. The authors in \cite{wu} constructed a quasi-interpolation scheme based on a generator that fails to satisfy the classical Strang--Fix conditions. They further introduced a generalized version of these conditions and proved that, by appropriately choosing a parameter depending on the sampling step, the corresponding quasi-interpolation operator satisfies an error estimate with the desired rate of convergence. The authors in \cite{gao} discussed the construction of radial kernels satisfying higher-order generalized Strang-Fix conditions. On the other hand, Maz'ya and his collaborators developed the theory of approximate approximations using generators that do not possess polynomial reproduction properties \cite{lazara1,lazara2,lazara3, mazya3, mazya4, mazya5}. However, these conditions are suitable for pointwise sampling and do not directly apply to approximation operators based on local average samples, including sampling Kantorovich operators. 

In \cite{CostarelliVinti}, the authors extensively investigated the approximation behavior of sampling Kantorovich operators by deriving approximation error estimates, analyzing regularization properties, and employing Strang--Fix type conditions to characterize convergence and approximation orders. However, to the best of our knowledge, no previous work has explicitly recognized the connection between these Strang--Fix type conditions and Bernoulli numbers. This observation motivates us to introduce the Bernoulli--Strang--Fix (BSF) conditions and their generalized versions as a natural framework for studying generators associated with sampling Kantorovich operators. These conditions enable us to derive desirable approximation orders for the sampling Kantorovich operator without requiring polynomial reproduction.

The problem of predicting a signal from its past samples traces its origins to the works of  Wainstein and Zubakov \cite{Wainstein} and Brown \cite{Brown}. Splettst"o{\ss}er subsequently advanced this field by developing its theoretical foundations and practical methodologies for signal prediction. Mugler and Splettst"o{\ss}er further extended these ideas in a series of publications \cite{Mugler1, Mugler2, Mugler3, Splettstoeser1, Splettstoeser2}, where they presented detailed analyses and significant extensions of the early prediction models. Later, Butzer \textit{et al.} developed the theory of signal prediction by establishing important theoretical results and demonstrating practical applications \cite{Butzer0, Butzer3, Butzer4, BBSV2}. In particular, they provided a comprehensive survey of the prediction problem in \cite{Butzer4}, which offers valuable insights into the development of the subject. More recently, the authors in \cite{Selvan1} investigated the convergence, approximation properties, and signal prediction from past samples using sampling expansions in shift-invariant spaces generated by smooth functions.

Periodic nonuniform sampling (PNS) sets play a fundamental role in the signal prediction problem. A PNS set is a sampling sequence of the form
$X=\{x_n+\rho l:\;0\le n\le L-1,\; l\in\mathbb{Z}\},$
where $\rho>0$ is the period of the sampling sequence, $L$ is the number of sampling locations within each period, and $(x_0,x_1,\ldots,x_{L-1})$ is the offset vector specifying the sampling pattern. The concept of periodic nonuniform sampling dates back to the works of Kohlenberg \cite{Kohlenberg} and Yen \cite{yen}, and has since found applications in channel noise reduction \cite{Aragones}, multiband signal analysis \cite{Herley}, reconstruction of bandlimited signals from function and derivative samples \cite{Nathan1973}, and finite-sample reconstruction with root-exponential accuracy \cite{StrohmerTanner}.
In our recent work \cite{SreyaGhoshSelvan}, we developed a framework for 
signal prediction through the perfect reconstruction of signals from shift-invariant spaces using periodic nonuniform samples of both the signal and its derivatives. In this work, we introduce the BSF conditions for a vector-valued generator $\varphi=(\varphi_0,\dots,\varphi_{\rho-1})$ and a PNS set $X$. We then investigate signal prediction through sampling Kantorovich operators associated with a pair $(\varphi, X)$. This extension establishes a suitable framework for analyzing signal reconstruction and prediction using sampling Kantorovich operators.

Maz'ya in \cite{mazya1, mazya2} introduced the concept of approximate quasi-interpolation, which provides the construction of approximation schemes using generators without requiring polynomial reproduction properties. Inspired by this concept, we introduce an analogous framework for the Kantorovich operator associated with the pair $(\varphi, X)$. The corresponding approximate sampling Kantorovich operator is defined by
\begin{equation*}
(K_{W,d}^{\varphi,X}f)(t)
=
\frac{1}{d}
\sum_{l\in\mathbb Z}
\sum_{n=0}^{\rho-1} W
\left[
\int\limits_{\frac{x_n+\rho l}{W}}^{\frac{x_n+\rho l+1}{W}}f(y)\,dy
\right]
\varphi_n\!\left(
\frac{Wt-\rho l}{d}
\right), \quad t \in \mathbb{R},
\end{equation*}
where $W,d \geq1$ are positive parameters.
In the particular case when $X=\mathbb{Z}$, $\rho=1$, and $d=1$, this operator reduces to the classical sampling Kantorovich operator defined in \eqref{K1butzer}. 
In contrast to conventional approximation processes, where convergence is typically governed by a single approximation parameter, the present framework involves two parameters, $W$ and $d$. As these parameters increase, the approximating function may converge to the original function, offering a more flexible and generalized approach to function approximation.

The major objectives of this paper are
outlined with each contribution presented separately:
\begin{itemize}[leftmargin=1.7em]
\item [($i$)]\textbf{Generalized Bernoulli--Strang--Fix Conditions:}
We first introduce the BSF and their generalized versions for a vector-valued generator $\varphi=(\varphi_0,\dots, \varphi_{\rho-1})$ and a PNS set $X$. We then examine their connection with vanishing moment conditions, which play a crucial role in analyzing the convergence of the sampling Kantorovich operator $K_{W,d}^{\varphi, X}$.  We further present a constructive method for obtaining a pair $(\varphi, X)$ satisfying the classical and generalized BSF conditions.

\item[($ii$)] \textbf{Approximation and Convergence Analysis:} 
Using the generalized BSF conditions, we establish the exact and asymptotic reproducing polynomial properties up to a prescribed degree for the sampling Kantorovich operators associated with $(\varphi, X)$. By combining these properties with Taylor expansions, we obtain the error estimate $\|K^{\varphi, X}_{W,d}f-f\|_p$ for functions $f$ belonging to Sobolev spaces; see Theorem \ref{approximation theorem} and Theorem \ref{theorem3.2}. Furthermore, we investigate the suitable choice of $d$ depending on $W$ for which the approximation error converges to zero as $W \to \infty$; see Corollary \ref{Loperator}.
 
\item[($iii$)] \textbf{Prediction from Local Average Samples:} We investigate the prediction of signals from a finite number of past local average samples using the proposed sampling Kantorovich operators $K_W^{\varphi, X}$; see Theorem \ref{prediction}. This demonstrates that the sampling Kantorovich operators can also be interpreted as prediction operators.

\item[($iv$)] \textbf{Numerical Validation:} Finally, we provide numerical examples using Gaussian generators and translates of B-splines to validate the proposed approximation and prediction results. These examples highlight the effectiveness and accuracy of the developed approximation and prediction operators.
\end{itemize}
The paper is structured as follows. Section 2 introduces the classical and generalized BSF conditions within the PNS framework. We establish their equivalent moment conditions and present a constructive approach for obtaining generators satisfying these conditions. Section 3 addresses the approximation and prediction analysis of the sampling Kantorovich operators, supported by illustrative examples.

\section{Generalized Bernoulli--Strang--Fix Conditions}
We begin by recalling some basic notation and terminology that will be used throughout the paper. 

Let $AC_{loc}^{(r)}(\mathbb{R})$ denote the space of all $r$-fold locally absolutely continuous functions on $\mathbb{R}$, and let $C_c^{\infty}(\mathbb{R})$ denote the space of all infinitely differentiable compactly supported functions on $\mathbb{R}.$  We assume that $p$ is a real number such that $1\leq p\leq \infty.$
For $r\in\mathbb{N}$,  the Sobolev space $\mathcal{W}^r_p(\mathbb{R})\equiv \mathcal{W}^r\left(L^p(\mathbb{R})\right)$ is given by
$$\mathcal{W}_p^r(\mathbb{R}):=\left\{f\in L^p(\mathbb{R}):  f(t) =g(t) ~a.e., ~g\in AC^{(r)}_{loc}(\mathbb{R}), ~g^{(r)}\in L^p(\mathbb{R})\right\}.$$
For $m\in\mathbb{N}_0$,
$\Pi_m=\Pi_m(\mathbb{R})$
denotes the space of all polynomials of degree at most $m$ on $\mathbb{R}$.
A continuous function $f$ belongs to the Wiener space $\mathcal{W}(\mathbb{R})$ if $$\|f\|_{\mathcal{W}(\mathbb{R})}:=\sum_{n\in\mathbb{Z}}\max\limits_{x\in[0,1]}|f(x+n)|<\infty.$$
The Fourier transform of a function $f\in L^{1}(\mathbb{R})$ is defined by
$$
\widehat{f}(\omega)
=\int\limits_{-\infty}^{\infty}f(x)e^{-2\pi i\omega x}\,dx,
\qquad \omega\in\mathbb{R}.
$$
If $f$ and $\widehat{f}$ belong to $\mathcal{W}(\mathbb{R})$, then the Poisson summation formula \cite{time} 
\begin{align*}\label{poissonsum}
\sum_{n\in\mathbb{Z}}f(x+\nu n)=\dfrac{1}{\nu}\sum_{n\in\mathbb{Z}}\widehat{f}\left(\dfrac{n}{\nu}\right)e^{2\pi \mathrm{i} nx/\nu},~\nu>0,
\end{align*}
holds for all $x\in\mathbb{R}$ with absolute convergence of both sums.

For $\epsilon>0$ and $r\in\mathbb{N}$, let $\mathcal{F}_{r,\epsilon}$ denote the class of generators $\varphi$ satisfying
$$
\varphi(t)=\mathcal{O}\left((1+|t|)^{-r-\epsilon-1}\right)
\quad \text{as } |t|\to\infty,
$$
and
$$
\widehat{\varphi}^{(s)}(\omega)
=\mathcal{O}\left((1+|\omega|)^{-r-\epsilon-1}\right)
\quad \text{ as } |\omega|\to\infty,
\quad  0\leq s\leq r-1.
$$
The first condition, together with the differentiation property of the Fourier transform, implies that $\widehat{\varphi}$ is $r-1$ times continuously differentiable on $\mathbb{R}$.

The Bernoulli numbers $B_n$ are recursively defined by
$$
B_0=1,\quad
B_n=-\frac{1}{n+1}\sum_{k=0}^{n-1}\binom{n+1}{k}B_k,\quad n\ge1.
$$
The Bernoulli numbers satisfy the following identity (see \cite{IrelandRosen}):
\begin{equation}\label{BK}
\sum_{k=0}^{j}\binom{j+1}{k}B_k=\delta_{j0},
\end{equation}
where $\delta_{j0}$ denotes the Kronecker delta.

A generator $\varphi\in\mathcal{F}_{r,\epsilon}$ is said to satisfy the \emph{Bernoulli--Strang--Fix $($BSF$)$ conditions of order $r$} if
\begin{equation*}\label{BSF1}
\widehat{\varphi}^{(j)}(l) =
(2\pi i)^j B_j \delta_{l0}, \quad j = 0,1,\dots,r-1, \quad l\in\mathbb{Z},
\end{equation*}
where $B_j$ denotes the $j$-th Bernoulli number.

We now extend this notion to the setting of a vector-valued generator 
$\varphi=(\varphi_0,\dots,\varphi_{\rho-1})\in \mathcal{F}_{r,\epsilon}^\rho$ together with a PNS set
$X=\{x_n+\rho l:0\le n\le \rho-1,\; l\in\mathbb Z\}.$

We say that a pair $(\varphi,X)$ satisfies the \emph{BSF conditions of order $r$} if
\begin{equation}\label{BSF4}
\sum_{n=0}^{\rho-1}\sum_{k=0}^{j}\binom{j}{k}(2\pi ix_n)^k
\widehat{\varphi_n}^{(j-k)}\left(\frac{l}{\rho}\right)
= \rho(2\pi i)^jB_j\delta_{l0},
\end{equation}
for $j=0,1,\ldots,r-1$ and $l \in \mathbb{Z}$. Observe that these conditions reduce to the classical BSF conditions when $\rho=1$ and $X=\mathbb{Z}$.

If the conditions \eqref{BSF4} hold for $l=0$ (not necessarily for non-zero integers) and $j=0,1,\dots,r-1$,
we say that $(\varphi,X)$ satisfies the \emph{generalized BSF conditions of order $r$}.

\begin{lemma}\label{BSF eqv1}
Let $\varphi\in\mathcal{F}_{r,\epsilon}^\rho$ and let
$
X=\{x_n+\rho l:\;0\le n\le \rho-1,\; l\in\mathbb Z\}$
be a PNS set. If $(\varphi,X)$  satisfies the generalized
BSF conditions of order $r$, then for $j=0,1,\ldots,r-1,$
$$
\frac{1}{d}
\sum_{l\in\mathbb Z}\sum_{n=0}^{\rho-1}
(x_n+\rho l-t)^j
\varphi_n\!\left(\frac{t-\rho l}{d}\right)
= B_j+\mathcal{O}\left(d^{j-r-\epsilon-1} \right) \text{ as }d \to \infty,$$
for every $t \in \mathbb{R}$, with the $\mathcal{O}(\cdot)$-constant independent of $t$.
\end{lemma}
\begin{proof}
Using the binomial expansion, we can write

\vspace{0.3cm}
$\ds\sum_{l \in \mathbb{Z}} \ds\sum_{n=0}^{\rho-1}
(x_n+ \rho l - t)^j \, \varphi_n\left(\frac{t- \rho l }{{d}}\right)$
\vspace{-0.3cm}
\begin{eqnarray}
\hspace{2cm}&=&
\sum_{l \in \mathbb{Z}}\sum_{n=0}^{\rho-1} \sum_{k=0}^j
\binom{j}{k} x_n^k (-t+\rho l)^{j-k}
\, \varphi_n\left(\frac{t-\rho l}{d}\right)
\nonumber\\
\hspace{2cm}&=&
\sum_{n=0}^{\rho-1} \sum_{k=0}^j
\binom{j}{k} x_n^k
\sum_{l \in \mathbb{Z}}
(-t+\rho l)^{j-k}
\, \varphi_n\left(\frac{t-\rho l}{d}\right).
\label{bino1}
\end{eqnarray}
For the function $f(x)=x^j\varphi_n\left(\dfrac{x}{d}\right)$, the Fourier transform is given by 
\begin{equation*}
\widehat{f}(\omega) = d \left(-2\pi i\right)^{-j} \widehat{\varphi}_n^{(j)}(d \omega).
\end{equation*}
Since $\varphi \in \mathcal{F}_{r,\epsilon}^\rho$, both $f$ and $\widehat{f}$ belong to $\mathcal{W}(\mathbb{R})$.
Now applying the Poisson summation formula to the inner sum over $l$ in \eqref{bino1}, we obtain

\vspace{0.3cm}
$\dfrac1{d}
\ds\sum_{l\in\mathbb Z}
\ds\sum_{n=0}^{\rho-1}
(x_n+\rho l-t)^j
\varphi_n\!\left(\frac{t-\rho l}{d}\right)$
\vspace{-0.3cm}
\begin{eqnarray}
&& \hspace{1cm} =
\frac1\rho
\sum_{l\in\mathbb Z}
\sum_{n=0}^{\rho-1}
\sum_{k=0}^{j}
\binom{j}{k}
(2\pi i x_n)^k
\left(-2\pi i\right)^{-j} 
\frac{d^{\,j-k}}{d\omega^{\,j-k}}
\left(
\widehat\varphi_n( d\,\omega)
\right)
\Bigg|_{\omega=\frac{l}{\rho}}
e^{2\pi i lt/\rho} \label{eqn1}\\
&& \hspace{1cm} =
B_j+\frac1\rho \sum_{\substack{l\in\mathbb Z\\l\neq0}} A_{l,j}(d) e^{2\pi i lt/\rho}, \label{eqn2}
\end{eqnarray}
where
$$
A_{l,j}(d)
=
\sum_{n=0}^{\rho-1}
\sum_{k=0}^{j}
\binom{j}{k}
(2\pi i x_n)^k \left(-2\pi i\right)^{-j} 
\frac{d^{\,j-k}}{d\omega^{\,j-k}}
\left(
\widehat\varphi_n(d\,\omega)
\right)
\Bigg|_{\omega=\frac{l}{\rho}}.
$$
By using the decay estimate of $\widehat{\varphi_n}$, we have
\begin{eqnarray*}
|A_{l,j}(d)|
&\leq&
C\sum_{n=0}^{\rho-1}
\sum_{k=0}^{j}
\binom{j}{k}
(2\pi  |x_n|)^k  \left({2 \pi }\right)^{-j}
 d^{j-k}
\left(\frac{\rho}{\rho+d|l|}\right)^{r+\epsilon+1}\\
&\leq&
C d^{j-r-\epsilon-1}|l|^{-r-\epsilon-1}, \quad \text{ for } l \in \mathbb{Z} \text{ and } j=0,\dots, r-1,
\end{eqnarray*}
where $C$  is a constant depending only  on $\rho,$ $j$, and the offset vector $(x_0,x_1,\ldots,x_{\rho-1})$. 
Hence, the desired result follows immediately from \eqref{eqn2}.
\end{proof}
The following result is an immediate consequence of \eqref{eqn1} and the uniqueness of Fourier series.
\begin{lemma}\label{BSF_equiv}
Let $\varphi\in\mathcal{F}_{r,\epsilon}^\rho$ and let
$
X=\{x_n+\rho l:\;0\le n\le \rho-1,\; l\in\mathbb Z\}$
be a PNS set. Then $(\varphi,X)$  satisfies the
BSF conditions of order $r$  if and only if
\begin{equation*}\label{BSF3}
\sum_{l\in\mathbb Z}\sum_{n=0}^{\rho-1}
(x_n+\rho l-t)^j\,\varphi_n(t-\rho l)
=B_j,
\quad j=0,1,\ldots,r-1.
\end{equation*}
\end{lemma}

In \cite{SreyaGhoshSelvan}, several  generators $\theta=(\theta_0, \dots, \theta_{\rho-1})$ and PNS sets $
X=\{x_n+\rho l:0\le n\le \rho-1,\; l\in\mathbb Z\}
$ are constructed such that, for each $j=0,1,\dots,r-1,$
\begin{equation}\label{vanishingmoment1}
\sum\limits_{l\in\mathbb{Z}}\sum\limits_{n=0}^{\rho-1}\left(x_n+\rho l-t\right)^{j}\theta_{n}(t-\rho l)=\delta_{j0},\quad t\in \mathbb{R}.
\end{equation}
We now choose the coefficients $b_k$, $k=0, \dots, r-1,$ such that \begin{equation}\label{b_k}
\sum_{k=0}^{r-1}b_k (-k)^j =B_j,\text{ for } j=0, \dots, r-1.
\end{equation}
The existence of such coefficients follows from the fact that these equations form a Vandermonde system. Using these coefficients, define a new generator $\varphi=(\varphi_0, \dots, \varphi_{\rho-1})$ by
\begin{equation}\label{phi_n}
\varphi_n(t)=\sum_{k=0}^{r-1}b_k \theta_n(t-k), \quad  n=0,\dots, \rho-1.
\end{equation}
A direct calculation shows that the pair $(\varphi, X)$ satisfies the BSF conditions of order $r$.

Although this construction provides a method for obtaining a pair $(\varphi, X)$ satisfying the BSF conditions, it still requires the verification of an infinite number of equations. In contrast, the generalized BSF conditions require verifying only a finite number of conditions, making this approach more practical and effective. To construct a pair $(\varphi, X)$ satisfying the generalized BSF conditions, we first choose a real-valued function $g \in \mathcal{F}_{r, \epsilon}$ and define
\begin{equation}\label{phin}
\varphi_n(t)=\sum_{m=0}^{r-1}b_m g(t-m-n), \quad n=0,\dots, \rho-1,  
\end{equation}
where $b_m\in\mathbb{R}$ are coefficients to be determined.
Substituting the Fourier transforms of $\varphi_n$ into the following equations, 
\begin{equation}\label{system2}
\sum_{n=0}^{\rho-1}\sum_{k=0}^{j}\binom{j}{k}(2\pi ix_n)^k
\widehat{\varphi_n}^{(j-k)}(0)
= \rho(2\pi i)^jB_j, \text{ for }j=0,\dots,r-1,
\end{equation}
results in a system of $r$ linear equations. 
Solving this system determines the coefficients $b_m$, $m=0,\dots,r-1,$ and hence the resulting pair $(\varphi, X)$ satisfies the generalized BSF conditions.

\section{Approximation and Prediction by Sampling Kantorovich Operators}
Let $\varphi\in\mathcal{F}_{r,\epsilon}^\rho$ and let
$
X=\{x_n+\rho l:\;0\le n\le \rho-1,\; l\in\mathbb Z\}$
be a PNS set. Recall that the approximate sampling Kantorovich operator associated with the pair $(\varphi, X)$ is defined by
\begin{equation}\label{K4}
(K_{W,d}^{\varphi,X}f)(t)
=
\frac{1}{d}
\sum_{l\in\mathbb Z}
\sum_{n=0}^{\rho-1} W
\left[
\int\limits_{\frac{x_n+\rho l}{W}}^{\frac{x_n+\rho l+1}{W}}f(y)\,dy
\right]
\varphi_n\!\left(
\frac{Wt-\rho l}{d}
\right), \quad t \in \mathbb{R},
\end{equation}
where $W,d \geq1$ are positive parameters.
For simplicity, we write $K_d^{\varphi, X}$ instead of $K_{1,d}^{\varphi,  X}$ and $K^{\varphi, X}$ instead of $K_{1,1}^{\varphi, X}$.
\begin{lemma}\label{RPP1}
Let $\varphi\in\mathcal{F}_{r,\epsilon}^\rho$ and let
$
X=\{x_n+\rho l:\;0\le n\le \rho-1,\; l\in\mathbb Z\}$
be a PNS set. If $(\varphi,X)$ satisfies the generalized
BSF conditions of order $r$, then for every $p \in \Pi_{r-1}$, the sampling Kantorovich operator
$K_{d}^{\varphi,X}$
satisfies
\begin{equation*}\label{Rpp3}
(K_{d}^{\varphi,X}p)(t)
=
p(t)
+
\mathcal{O}\!\left(
d^{-2-\epsilon}
\sum_{j=0}^{r-1}
\Big|p^{(j)}(t)\Big|\right)
\text{ as } d\to\infty,  
\end{equation*}
for every $t \in \mathbb{R}$, with the $\mathcal{O}(\cdot)$-constant independent of $t$.
\end{lemma}
\begin{proof}
Let $p$ be a polynomial of degree less than $r$. Then its Taylor expansion about any point $t$ is exact (with no remainder term). Hence, for every $t,y \in \mathbb{R}$,
\begin{equation}\label{TF2}
p(y) = \sum_{j=0}^{r-1} \frac{p^{(j)}(t)}{j!}(y - t)^j.
\end{equation} 
Substituting \eqref{TF2} into \eqref{K4}, we obtain $(K_d^{\varphi,X}p)(t)$
\begin{eqnarray}\label{K_estimate}
&=&
\frac{1}{ d}\sum_{l \in \mathbb{Z}} \sum_{n=0}^{\rho-1}
\left[
\int\limits_{x_n+\rho l}^{x_n+\rho l+1}
\sum_{j=0}^{r-1}
\frac{p^{(j)}(t)}{j!}(y-t)^j
\,dy
\right]
\varphi_n\!\left( 
\frac{t-\rho l}{d}\right) \nonumber
\end{eqnarray}
\begin{eqnarray}
\hspace{1cm}&=&
\frac{1}{d}\sum_{l \in \mathbb{Z}} \sum_{n=0}^{\rho-1}
\sum_{j=0}^{r-1}
\frac{p^{(j)}(t)}{j!}
\int\limits_{x_n+\rho l}^{x_n+\rho l+1}
(y-t)^j\,dy 
\;\varphi_n\!\left(
\frac{t-\rho l}{d}\right)\nonumber\\
&=&
\frac{1}{d}\sum_{l\in \mathbb{Z}} \sum_{n=0}^{\rho-1}
\sum_{j=0}^{r-1}
\frac{p^{(j)}(t)}{(j+1)!}
\left[
(x_n+\rho l+1-t)^{j+1}
-(x_n+\rho l-t)^{j+1}
\right]
% \\
% &&\hspace{7.5cm}\times~
\varphi_n\!\left(
\frac{t-\rho l}{d}\right) \nonumber\\
&=&
\sum_{j=0}^{r-1}
\frac{p^{(j)}(t)}{(j+1)!}
\sum_{k=0}^{j}
\binom{j+1}{k}
\frac1{d}
\sum_{l\in\mathbb Z}
\sum_{n=0}^{\rho-1}
(x_n+\rho l-t)^k
\varphi_n\!\left(
\frac{t-\rho l}{d}
\right).
\end{eqnarray}
The last equality follows from the identity
$
(t+1)^m-t^m
=
\ds\sum_{k=0}^{m-1}
\binom{m}{k}t^k.
$
Using Lemma~\ref{BSF eqv1} and \eqref{BK}, we have $\left|(K_{d}^{\varphi,X}p)(t)-p(t)\right|$
\begin{eqnarray}\label{eq3.7}
&\leq&
\sum_{j=0}^{r-1}
\frac{|p^{(j)}(t)|}{(j+1)!}
\sum_{k=0}^{j}
\binom{j+1}{k}\left|\frac1{d}
\sum_{l\in\mathbb Z}
\sum_{n=0}^{\rho-1}
(x_n+\rho l-t)^k
\varphi_n\!\left(
\frac{t-\rho l}{d}
\right)-B_k \right|
\nonumber\\
&\leq& 
C\sum_{j=0}^{r-1}
\frac{|p^{(j)}(t)|}{(j+1)!}
\sum_{k=0}^{j}
\binom{j+1}{k}
d^{k-r-\epsilon-1}. \nonumber  
\end{eqnarray}
This completes the proof.
\end{proof}
The previous lemma establishes an asymptotic polynomial reproduction property under the generalized BSF conditions.
If $(\varphi,X)$ satisfies the BSF conditions, then it follows from \eqref{K_estimate} and Lemma \ref{BSF_equiv} that the sampling Kantorovich operator reproduces polynomials exactly. This is stated in the following lemma.
\begin{lemma}\label{RPP2}
Let $\varphi\in\mathcal{F}_{r,\epsilon}^\rho$ and let
$X=\{x_n+\rho l:\;0\le n\le \rho-1,\; l\in\mathbb Z\}$
be a PNS set. If $(\varphi,X)$ satisfies the BSF conditions of order $r$, then the sampling Kantorovich operator  $K^{\varphi,X}$ satisfies the reproducing polynomial property of order $r$, \textit{i.e.,} 
$$ K^{\varphi,X}p=p, \text{ for every } p\in\Pi_{r-1}.$$
\end{lemma}

The sampling Kantorovich operator $K_d^{\varphi,X}$ can be represented as an integral operator of the form
 \begin{align}\label{K_integral}
(K_d^{\varphi,X}f)(t)
&=
\int\limits_{-\infty}^{\infty}
\mathcal{N}_d(t,y)f(y)\,dy, \quad t \in \mathbb{R},
\end{align}
where the kernel $\mathcal{N}_d$ is given by
$$
\mathcal{N}_d(t,y)
=
\frac{1}{d}\sum_{l \in \mathbb{Z}}
\sum_{n=0}^{\rho -1}
\chi_{[0,1)}(y-x_n-\rho l)
\varphi_n\left(\frac{t-\rho l}{d}\right).
$$
The kernel $\mathcal{N}_d$ satisfies the following properties:
\begin{itemize}
\item[($i$)]
$\mathcal{N}_d\!\left(t-\rho \nu,y\right)
=\mathcal{N}_d\!\left(t,y+\rho \nu\right),
\text{ for all } \nu\in\mathbb Z.$
\item[($ii$)]
$\int\limits_{-\infty}^{\infty}
|\mathcal{N}_d(t,\cdot)|\,dt
\in L^\infty[0,\rho).$
\item[($iii$)]
$\int\limits_{-\infty}^{\infty}
(\rho+|y|)^r
|\mathcal{N}_d(\cdot,y)|\,dy
\in L^\infty[0,\rho)$.
\end{itemize}
Property ($i$) is immediate from the definition of the kernel. 
To prove  ($ii$) we use a change of variable together with the partition of unity of the characteristic function of the interval $[0,1)$. Indeed,
\begin{eqnarray}\label{N_d1}
\int\limits_{-\infty}^{\infty} |\mathcal{N}_d(t,y)|\,dt
&\le&
\int\limits_{-\infty}^{\infty}
\left|
\frac{1}{d}
\sum_{l\in\mathbb Z}
\sum_{n=0}^{\rho-1}
\chi_{[0,1)}(y-x_n-\rho l)
\varphi_n\!\left(\frac{t-\rho l}{d}\right)
\right|
dt \nonumber \\
&\le&
\sum_{l\in\mathbb Z}
\sum_{n=0}^{\rho-1}
\chi_{[0,1)}(y-x_n-\rho l)
\int\limits_{-\infty}^{\infty}
|\varphi_n(t)|\,dt =N < \infty,
\end{eqnarray}
where  the constant $N$ is independent of $d$.
Finally, to prove ($iii$), define
$$
G_d(t)
=
\int\limits_{-\infty}^{\infty}
(\rho +|y|)^r
|\mathcal{N}_d(t,y)|\,dy, \quad t\in[0, \rho).
$$
Then it remains to show that $G_d \in L^\infty[0,\rho)$. By the decay estimate of $\varphi_n$, we have
\begin{eqnarray}\label{Gd estimate}
G_d(t)
&\leq& \frac{1}{d}
\int\limits_{-\infty}^{\infty}
(\rho+|y|)^r
\sum_{l\in\mathbb Z}
\sum_{n=0}^{\rho-1}
\chi_{[0,1)}(y-x_n-\rho l)
\left|\varphi_n\left(\frac{t-\rho l}{d}\right)\right|
\,dy \nonumber \\
&\leq&
 \frac{C}{d}
\int\limits_{-\infty}^{\infty}
(\rho+|y|)^r
\sum_{l\in\mathbb Z}
\sum_{n=0}^{\rho-1}
\chi_{[0,1)}(y-x_n-\rho l)
\left(1+\frac{|t-\rho l|}{d}\right)^{-r-\epsilon-1}
\,dy \nonumber
\\
&=&
 C d^{r+\epsilon}
\sum_{l\in\mathbb Z}
\sum_{n=0}^{\rho-1}
\int\limits_{x_n+\rho l}^{x_n+\rho l+1}
(\rho+|y|)^r \left(d+|t-\rho l|\right)^{-r-\epsilon-1}\, dy.
\end{eqnarray}
When
$
y\in
\left[
x_n+\rho l,
x_n+\rho l+1
\right],
$
we have
$
|y| \leq |x_n+\rho l|+1.$ Hence, by the triangle inequality, there exists a constant $C'>0$ (depending only on $\rho$) such that
\begin{equation}\label{eqnforg1}
 (\rho+|y|)
\leq
C'(1+|l|), \text{ for } y\in
\left[
x_n+\rho l,
x_n+\rho l+1
\right] \text{ and } l \in \mathbb{Z}.
\end{equation}
Since
$t\in[0,\rho)$, by the reverse triangle inequality,
\begin{equation}\label{eqnforg2}
d+|t-\rho l|
\ge d+\rho|l|-\rho \ge \frac{\rho}{2}(1+|l|),
\text{ for } |l|\ge 3, ~ d \geq 1, \text{ and } t\in[0, \rho).
\end{equation}
Substituting \eqref{eqnforg1} and \eqref{eqnforg2} in \eqref{Gd estimate}, we get
\begin{eqnarray}\label{G_d estimate}
G_d(t) &\leq&
C' d^{r+\epsilon} \Bigg[\sum_{n=0}^{\rho-1}
\sum_{|l|< 3}
(1+|l|)^r
\Big[\left(d+|t-\rho l|\right)^{-r-\epsilon-1} \nonumber\\
&& \hspace{3cm} + \frac{\rho^2}{2} \sum_{|l|\geq 3}
(1+|l|)^{-r-\epsilon-1}\Big]\Bigg] \nonumber\\ 
&\leq& 
C_1d^{r+\epsilon}, \text{ for } t \in [0,\rho).
\end{eqnarray}
Thus, $G_d \in L^\infty[0,\rho)$.
Similarly,  we can show that there exists a constant $C_2>0$ independent of $d$ such that
\begin{equation}\label{N_d2}
\int\limits_{-\infty}^{\infty} |\mathcal{N}_d(t,y)|\,dy \leq C_2 d^{r+\epsilon}, \quad t \in [0,\rho).
\end{equation}
We now estimate $\|K^{\varphi, X}_d\|_p $.  
By applying H$\Ddot{o}$lder's inequality  with $\tfrac{1}{p}+\tfrac{1}{q}=1$ in \eqref{K_integral}, we obtain from \eqref{N_d2} that
\begin{align*}
\left|(K_d^{\varphi,X}f)(t)\right|^p
&=
\left(
\int\limits_{-\infty}^{\infty}
|\mathcal{N}_d(t,y)|\,|f(y)|^p\,dy
\right)
\left(
\int\limits_{-\infty}^{\infty}
|\mathcal{N}_d(t,y)|\,dy
\right)^{p-1}\\
&\le
C_2^{p-1}\left(
d^{r+\epsilon}
\right)^{p-1}
\int\limits_{-\infty}^{\infty}
|\mathcal{N}_d(t,y)|\,|f(y)|^p\,dy.
\end{align*}
Applying \eqref{N_d1}, we obtain the following estimate
\begin{align*}
\left\|K_d^{\varphi,X}f\right\|_p^p
&\le C_2^{p-1}
\left(
d^{r+\epsilon}
\right)^{p-1}
\int\limits_{-\infty}^{\infty}
\int\limits_{-\infty}^{\infty}
|\mathcal{N}_d(t,y)|\,|f(y)|^p
\,dy\,dt \\
&=
C_2^{p-1}\left(
d^{r+\epsilon}
\right)^{p-1}
\int\limits_{-\infty}^{\infty}
|f(y)|^p
\left(
\int\limits_{-\infty}^{\infty}
|\mathcal{N}_d(t,y)|\,dt
\right)
dy \\
&\le N C_2^{p-1}
\left(
d^{r+\epsilon}
\right)^{p-1}
\|f\|^p_p.
\end{align*}
Let  $M=\max\{N,C_2\}$.  Then
\begin{equation}\label{Kbound}
    \left\|K^{\varphi,X}_df\right\|_p \leq M\left(d^{r+\epsilon}
\right)^{1-\frac{1}{p}} \|f\|_p \leq M\
d^{r+\epsilon} \|f\|_p.\\
\end{equation}

The kernel $\mathcal{N}_d$ satisfies the assumptions analogous to those in Theorem 2.1 of \cite{LeiJiaCheney}, where the case $\rho=d=1$  is considered. Following their approach, we aim to establish the error estimate $\|K^{\varphi, X}_{W,d}f-f\|_p$ for functions $f$ belonging to $\mathcal{W}_p^r(\mathbb{R})$.

Let $\gamma \in C_c^\infty(\mathbb{R})$ be such that $\int\limits_{-\infty}^{\infty} \gamma(x) \, dx =1$.
Consider the following smoothing operator $J$ defined by 
$$
(Jf)(x)
=
\int\limits_{-\infty}^{\infty} \big[f(x)-\nabla_u^r f(x)\big] \, \gamma(u)\,du,
\quad x\in \mathbb{R},
$$
where $r$ is a positive integer and $\nabla_u^r$ is the $r$-th difference operator defined by
\[
\nabla_u^r = (I-T_u)^r,
\quad u\in \mathbb{R},
\]
with $T_u$ denoting the translation operator and $I$ denoting the identity operator. The following lemma is proved in \cite{jia}.
\begin{lemma}\label{lemma2}  \cite{jia}
The operator $J$ maps locally integrable functions to finitely differentiable functions. Moreover, there is a constant $C$ independent of $p$ and $f\in L^p(\mathbb{R})$ such that

\begin{itemize}
    \item [($i$)] $\|Jf\|_p \leq C\|f\|_p$
     \item [($ii$)] $ \left\|(Jf)|_{\mathbb{Z}}\right\|_{\ell_p}
    \leq C\|f\|_p$
      \item [($iii$)] $ \|Jf-f\|_p
    \leq C\left\|f^{(r)}\right\|_{p},
     \text{ for every } f\in \mathcal{W}_p^r(\mathbb{R})$
       \item [($iv$)] $ Jq=q,
    \text{ for every } q\in \Pi_{r-1}.$
\end{itemize}
\end{lemma}
\begin{theorem}\label{approximation theorem}
Let $\varphi\in\mathcal{F}_{r,\epsilon}^\rho$ and let
$
X=\{x_n+\rho l:\;0\le n\le \rho-1,\; l\in\mathbb Z\}$
be a PNS set. If $(\varphi,X)$  satisfies the generalized BSF conditions of order $r$, then for every $f \in \mathcal{W}^{r}_p(\mathbb{R})$,
$$
\left\|K^{\varphi, X}_{W,d}f-f\right\|_p
\le
C_1\,W^{-r}d^{r+\epsilon}\left\|f^{(r)}\right\|_p +C_2 \,W^{-r}\,\left\|f^{(r)}\right\|_p +C_3\,d^{-2-\epsilon}
\sum_{j=0}^{r-1}W^{-j}\left\|f^{(j)}\right\|_p,
$$ where the constants $C_1$, $C_2$, and $C_3$ are independent of $f$, $d$, $p$, and $W$.
\end{theorem}
\begin{proof}
We prove the result for $W=1$. The general case follows by a change of variables. Let $f \in \mathcal{W}^r_p(\mathbb{R})$.
 From \eqref{Kbound} and part ($iii$) of Lemma~\ref{lemma2}, we get
\begin{eqnarray}\label{first_term0}
\left\|K^{\varphi,X}_df-f\right\|_p
&\le&
\left\|K^{\varphi,X}_df-K^{\varphi,X}_dJf\right\|_p+\left\|K^{\varphi,X}_dJf-Jf\right\|_p+\left\|Jf-f\right\|_p \nonumber\\
&\leq&
\left(M d^{r+\epsilon}+1\right)C\,\left\|f^{(r)}\right\|_p
+\left\|K^{\varphi,X}_dJf-Jf\right\|_p.
\end{eqnarray}
We now prove that
\begin{equation}\label{KJ_bound}
\left\|K^{\varphi, X}_dJf-Jf\right\|_p
\le C_1  d^{r+\epsilon}\left\|f^{(r)}\right\|_{p} 
+ C_3 d^{-2-\epsilon}
\sum_{j=0}^{r-1}
\left\|f^{(j)}\right\|_{p}.
\end{equation}
We prove the result for $1 \leq p < \infty$. The case $p=\infty$ follows analogously by replacing the integral with the supremum throughout the argument.

Set
$
g_d=K^{\varphi,X}_dJf-Jf
$
and define
$
a_{d,t}(\nu)=g_d(t- \rho \nu),$ for $\nu\in\mathbb{Z}.$
Since $g_d\in L^p(\mathbb{R})$, we have
\begin{equation}\label{K_norm}
\left\|K_d^{\varphi,X}Jf-Jf\right\|_p^p=\int\limits_{-\infty}^{\infty} |g_d(t)|^p\,dt  
=\sum_\nu \int\limits_0^\rho |g_d(t-\rho \nu)|^p\,dt 
=\int\limits_0^\rho \left\|a_{d,t}\right\|_{\ell_p}^p\,dt .
\end{equation}
We now estimate $\|a_{d,t}\|_{l_p}$, for $t \in [0, \rho)$.
Let $q_z$ be the $(r-1)$-th Taylor polynomial of $Jf$
about $z\in\mathbb R$, \textit{i.e.,}
$$q_z(t)=\sum_{j=0}^{r-1} \frac{(Jf)^{(j)}(z)}{j!}(t-z)^j,$$ and let $r_z$ be the corresponding remainder, $Jf-q_z$. Note that we have $Jf(z)=q_z(z)$, for every $z\in \mathbb{R}.$

Let $E_{d,q_z}(t)=K_d^{\varphi,X}q_z(t)-q_z(t)$. By  property $(i)$ of the kernel $\mathcal{N}_d$, the operator
$K^{\varphi,X}_d$ commutes with translations by $\rho\mathbb{Z}$. Using this fact, we obtain
\begin{eqnarray*}\label{a_{d,t}_2}
a_{d,t}(\nu)
&=&
g_d(t-\rho \nu) = (K^{\varphi, X}_dJf-Jf)(t-\rho \nu)  \nonumber \\
&=&
\left(K^{\varphi, X}_d(Jf-q_{t-\rho \nu})\right)(t-\rho \nu)+E_{d,q_{t-\rho \nu}}(x-\rho \nu) \nonumber \\
&=&
\left(K^{\varphi, X}_d r_{t-\rho \nu}\right)(t-\rho \nu)+E_{d,q_{t-\rho \nu}}(t- \rho \nu) \nonumber\\
&=&
K^{\varphi, X}_d T_{\rho \nu} r_{t-\rho \nu}(t)+E_{d,q_{t-\rho \nu}}(t-\rho \nu) \nonumber \\
&=&
\int\limits_{-\infty}^{\infty}
\mathcal{N}_d(t,y)e_{t,y}(\rho \nu)\,dy
+E_{d,q_{t-\rho \nu}}(t-\rho \nu),
\end{eqnarray*}
where
$
e_{t,y}(\nu)
=
r_{t-\nu}(y-\nu)
=
(T_\nu r_{t-\nu})(y).
% =
% \int\limits_0^1
% \frac{(1-s)^{r-1}}{(r-1)!}
% \,T_{-t-s(y-x)}
% \bigl(Jf\bigr)^{(r)}(-\nu)
% \,(y-t)^r\,ds.
$
Now using the inequality $|a+b|^p\leq 2^p \left(|a|^p+|b|^p\right)$, we have
$$
\left|a_{d,t}(\nu)\right|^p\leq 2^p\left[\int\limits_{-\infty}^{\infty} \left|\mathcal{N}_d(t,y)\right||e_{t,y}(\rho \nu)|\,dy\right]^p
+\left|E_{d,q_{t-\rho \nu}}(t-\rho \nu)\right|^p.$$ 
Applying the generalized Minkowski's inequality, we get
\begin{equation*}\label{a_d}
\left\|a_{d,t}\right\|_{\ell_p}\leq C\int\limits_{-\infty}^{\infty} \left|\mathcal{N}_d(t,y)\right|\left\|e_{t,y}(\rho \boldsymbol{\cdot})\right\|_{\ell_p}\,dy
+\left\|E_{d,q_{t-v}}(x- \rho \boldsymbol{\cdot})\right\|_{\ell_p}.
\end{equation*}
To estimate $\left\|e_{t,y}(\rho \boldsymbol{\cdot})\right\|_{\ell_p}$, we use the integral form of the remainder in
Taylor's theorem:
\begin{align*}
e_{t,y}(\nu)
&=r_{t-\nu}(y-\nu) =(Jf-q_{t-\nu})(y-\nu) \\
&=\int\limits_0^1
\frac{1}{(r-1)!}(Jf)^{(r)}(t- \nu+s(y-t))(1-s)^{r-1}(y-t)^r\,ds \\
&=\int\limits_0^1 \frac{1}{(r-1)!}
(T_{-t-s(y-t)} (Jf)^{(r)}(-\nu)
(1-s)^{r-1}(y-t)^r\,ds .
\end{align*}
Since the operators $J$, $T_\nu$, and the differential operator commute with one another, we have for any $t,y\in\mathbb{R}$ and $\nu\in\mathbb{Z}$,
\begin{equation}\label{2.7}
|e_{t,y}(\nu)|
\le
|y-t|^r
\int\limits_0^1
\left|
(JT_{-t-s(y-t)} f^{(r)}(-\nu)
\right|ds .
\end{equation}
Since $ f^{(r)} \in L^p(\mathbb{R})$, we obtain from Lemma \ref{lemma2} ($ii$) that
\begin{align}\label{2.8}
\left\|
(J(T_{-t-s(y-t)} f^{(r)})\big|_{\mathbb{Z}}
\right\|_{\ell_p}
\le
C
\left\|T_{-t-s(y-t)}f^{(r)}\right\|_p 
=
C\left\|f^{(r)}\right\|_{p}.
\end{align}
Combining \eqref{2.7} with \eqref{2.8} and using the generalized Minkowski
inequality,  we get
\begin{eqnarray*}
\left\|e_{t,y}\right\|_{\ell_p}
&\le&
|y-t|^r
\int\limits_0^1
\left\|(JT_{-t-s(y-t)} f^{(r)})\big|_{\mathbb{Z}}\right\|_{\ell_p}\,dt \\
&\le&
C|y-t|^r
\int\limits_0^1
\left\|f^{(r)}\right\|_{p}\,dt
= C|y-t|^r\left\|f^{(r)}\right\|_{p}.
\end{eqnarray*}
Hence
$$\left\|e_{t,y}(\rho \boldsymbol{\cdot})\right\|_{\ell_p}\leq \left\|e_{t,y}\right\|_{\ell_p} \leq C|y-t|^r\left\|f^{(r)}\right\|_{p}.$$
Using this estimate together with \eqref{G_d estimate}, we obtain for any $t \in [0, \rho)$, 
\begin{eqnarray} \label{first term_2}
\int\limits_{-\infty}^{\infty} \left|\mathcal{N}_d(t,y)\right|\left\|e_{t,y}(\rho \boldsymbol{\cdot})\right\|_{\ell_p}\,dy 
&\leq&
C\int\limits_{-\infty}^{\infty} |\mathcal{N}_d(t,y)|\,|t-y|^r\left\|f^{(r)}\right\|_{p}\,dy \nonumber
\end{eqnarray}
\begin{eqnarray}
&\leq&
C\left\|f^{(r)}\right\|_{p}
\int\limits_{-\infty}^{\infty} (\rho +|y|)^r |\mathcal{N}_d(t,y)|\,dy \nonumber\\
&\leq&
C_1d^{r+\epsilon}\left\|f^{(r)}\right\|_p.
\end{eqnarray}
Next we estimate $\left\|E_{d,q_{t-\rho \nu}}(t- \rho \boldsymbol{\cdot})\right\|_{\ell_p}$.
From Lemma \ref{RPP1}, we have
\begin{eqnarray*}
\left|E_{d,q_{t-\rho \nu}}(t-\rho \nu)\right|
&\le& C_3\,d^{-2-\epsilon}
\sum_{j=0}^{r-1}
\left|q_{t-\rho \nu}^{(j)}(t-\rho\nu)\right|\\
&\le&
C_3\,d^{-2-\epsilon}
\sum_{j=0}^{r-1}
\left|(Jf)^{(j)}(t-\rho\nu)\right|,
\end{eqnarray*}
where the constant $C_3$ is independent of $d$ and $t$.
Hence
\begin{eqnarray}\label{second term_2}
\left\|E_{d,q_{t-\rho \nu}}(t- \rho \boldsymbol{\cdot})\right\|_{\ell_p}
&\le&
C_3\,d^{-2-\epsilon}
\sum_{j=0}^{r-1}
\left\|(Jf)^{(j)}(t- \rho \boldsymbol{\cdot})\right\|_{\ell_p}\nonumber\\
&=&
C_3\,d^{-2-\epsilon}
\sum_{j=0}^{r-1}
\left\|JT_{-t}f^{(j)}(-\rho \boldsymbol{\cdot})\right\|_{\ell_p}\nonumber\\
&\le&
C_3\,d^{-2-\epsilon}
\sum_{j=0}^{r-1}
\left\|T_{-t}f^{(j)}\right\|_{p} =
C_3\,d^{-2-\epsilon}
\sum_{j=0}^{r-1}
\left\|f^{(j)}\right\|_{p}.
\end{eqnarray}
Combining the estimates \eqref{first term_2} and \eqref{second term_2}, and using \eqref{K_norm},  we obtain \eqref{KJ_bound}.
Finally, the desired result follows from \eqref{first_term0} and \eqref{KJ_bound}. 
\end{proof}
From Theorem \ref{approximation theorem}, we observe that the convergence of the approximation depends critically on the choice of the parameter $d$. Indeed, $d$ affects the error bound in two opposite ways: increasing $d$ decreases the last term of the error estimate, while simultaneously increasing the first term. Consequently, $d$ cannot be arbitrarily chosen; instead, it must be selected as a function of $W$ so that the total error goes to zero as $W\to\infty$. To guarantee convergence, it is sufficient to choose $d=d(W)$ such that
$$
W^{-r}d(W)^{r+\epsilon}\to 0  \text{ and } d(W)^{-2-\epsilon} \to 0
\text{ as }W\to\infty.$$
Under these conditions, each term in the error estimate converges to zero as $W\to\infty$.
A natural choice is
$
d(W)=W^\alpha,~ \alpha \geq0.
$
In this case, 
$$W^{-r}W^{\alpha\left(r+\epsilon\right)}\to 0 \text{ as } W\to\infty
\text { if and only if }
0\leq\alpha<\frac{r}{r+\epsilon}.
$$
Moreover,
$$
d^{-2-\epsilon}
=
W^{-\alpha(2+\epsilon)}
\to 0 \text{ as } W \to \infty, \text{ for every } \alpha\geq0. 
$$
Consequently, Theorem  \ref{approximation theorem} immediately yields the following convergent result. 
\begin{corollary}\label{Loperator}
Let $\varphi\in\mathcal{F}_{r,\epsilon}^\rho$ and let
$
X=\{x_n+\rho l:\;0\le n\le \rho-1,\; l\in\mathbb Z\}$
be a PNS set. 
Define the operator
$$
L_{W,\alpha}^{\varphi,X}:=K_{W, W^\alpha}^{\varphi,X},
\quad 0\leq \alpha<\frac{r}{r+\epsilon}.$$
If $(\varphi,X)$  satisfies the generalized
BSF conditions of order $r$, then 
$$
\lim_{W\to\infty}\left\|L_{W, \alpha}^{\varphi,X}f-f\right\|_p=0, \text{ for every } f\in \mathcal{W}_p^r(\mathbb R).
$$
\end{corollary}
Theorem \ref{approximation theorem} establishes an approximate approximation estimate under the generalized BSF conditions. When $d=1$ and $(\varphi, X)$ satisfies the BSF conditions,  Lemma \ref{RPP2} implies that the term $E_{d, q_z}$  appearing in the corresponding proof is equal to zero. Therefore, the last error term in Theorem \ref{approximation theorem} vanishes. As a consequence, we have the following result.  
\begin{theorem}\label{theorem3.2}
Let $\varphi\in\mathcal{F}_{r,\epsilon}^\rho$ and let
$
X=\{x_n+\rho l:\;0\le n\le \rho-1,\; l\in\mathbb Z\}$
be a PNS set. If  $(\varphi, X)$ satisfies the BSF conditions of order $r$, then 
$$
\left\|K^{\varphi,X}_Wf - f\right\|_p
\leq
C \,W^{-r} \, \left\|f^{(r)}\right\|_p, \text{ for every }f \in \mathcal{W}^{r}_p(\mathbb{R}),
$$
where $C>0$ is a constant independent of $f$, $p$, and $W$.
\end{theorem}

We now turn to the problem of signal prediction based on the sampling Kantorovich operator associated with $(\varphi, X)$. To predict a signal from a finite number of its past average samples of the signal based on this operator, it is necessary to assume that the generator $\varphi$ has compact support contained in the interval $(0, \infty)$. To construct such generators satisfying the BSF conditions of order $r$, we employ the following technique, which was introduced by the authors in \cite{SreyaGhoshSelvan}. 
\begin{lemma}\label{vander}\cite{SreyaGhoshSelvan}
  Consider the Vandermonde system
\begin{equation*} \label{VS}
   \sum_{p=0}^{\rho-1} a_p (-\varepsilon_p)^m = \delta_{m0},
\qquad m = 0,1,\ldots,\rho-1, 
\end{equation*}
where $\varepsilon_0 < \varepsilon_1 < \cdots < \varepsilon_{\rho-1}$ are given real numbers. The solution $(a_p)_{p=0}^{\rho-1}$ is unique and coincides with the Lagrange interpolation weights at \(x=0\):
\begin{equation}\label{a_p2}
  a_p = \prod_{\substack{q=0 \\ q\neq p}}^{\rho-1} \frac{\varepsilon_q}{\varepsilon_q - \varepsilon_p}, \quad p=0,1,\dots,\rho-1.  
\end{equation}  
\end{lemma}
Let us introduce the generator $\psi=(\psi_0, \dots, \psi_{\rho-1})$ by
$$
\psi_n(t) = \sum_{p=0}^{\rho -1} a_p \, \varphi_n(t - \varepsilon_p), \quad n=0,\dots,\rho-1,
$$
where the coefficients $a_p$ are precisely those given by Lemma \ref{vander}. 
Then the Fourier transforms of $\psi_n$ are given by
\begin{eqnarray*}
\widehat{\psi_n}(\omega)
= \widehat{\varphi_n}(\omega)\, p(\omega),
\text{ where } p(\omega) = \ds\sum_{p=0}^{\rho -1} a_p e^{-2\pi i \varepsilon_p \omega}.
\end{eqnarray*}
Note that $p^{(m)}(0)=(2 \pi i)^m\delta_{m0}.$
We now show that if $(\varphi,X)$ satisfies the BSF conditions of order $r$, then so does $(\psi, X).$
Indeed, for $j=0,1, \dots, r-1$, we have 
$$\ds\sum_{n=0}^{\rho-1}\sum_{k=0}^{j}\binom{j}{k}(2\pi ix_n)^k
\widehat{\psi_n}^{(j-k)}\left(\frac{l}{\rho}\right)=\ds\sum_{n=0}^{\rho-1} \frac{d^j}{d\omega^j}
\left( e^{2\pi i x_n(\omega - \frac{l}{\rho})} \widehat{\psi_n}(\omega) \right)
\Bigg|_{\omega = \frac{l}{\rho}}$$
\begin{eqnarray}\label{equation5.3}
\hspace{2cm}&=& \sum_{n=0}^{\rho-1} \frac{d^j}{d\omega^j}
\left( e^{2\pi i x_n(\omega - \frac{l}{\rho})} \widehat{\varphi_n}(\omega)\, p(\omega)\right)
\Bigg|_{\omega = \frac{l}{\rho}} \nonumber\\
\hspace{2cm}&=&\sum_{n=0}^{\rho-1} \sum_{m=0}^j \binom{j}{m} p^{(m)}\left(\frac{l}{\rho}\right) \frac{d^{j-m}}{d\omega^{j-m}} 
\left(
e^{2\pi i x_n\left(\omega - \frac{l}{\rho}\right)} 
\, \widehat{\varphi}_n(\omega)
\right)\Bigg|_{\omega = \frac{l}{\rho}}\nonumber\\
\hspace{2cm}&=& \sum_{m=0}^j \binom{j}{m} p^{(m)}\left(\frac{l}{\rho}\right) \rho (2 \pi i)^{j-m}B_{j-m} \delta_{l0} = \rho (2 \pi i)^j B_j \delta_{l0}.\nonumber
\end{eqnarray}
Let us introduce a new sampling Kantorovich operator
\begin{equation}\label{K3}
(K^{\psi,X}_Wf)(t)
=
\sum_{l\in\mathbb Z}
\sum_{n=0}^{\rho-1}
\left(
W\int_{\frac{x_n+\rho l}{W}}^{\frac{x_n+\rho l+1}{W}}
f(y)\,
dy
\right)
\psi_n(Wt-\rho l), \quad t \in \mathbb{R}.
\end{equation}
We assume that the functions $\varphi_n$ are compactly supported with support $[T_0, T_1]$. Then
$$\operatorname{supp}(\psi_n)=[T_0+\varepsilon_0, T_1+\varepsilon_{\rho -1}], \text{ for } n=0,\dots, \rho-1.$$
Consequently, the summation in \eqref{K3} extends only over those values of $l \in \mathbb{Z}$  for which
$$T_0+\varepsilon_0 \leq Wt- \rho l \leq T_1+\varepsilon_{\rho-1}.$$
This condition is equivalent to restricting $l$ to the set 
$$\Omega^n_t:=\left\{l \in \mathbb{Z} : t-\frac{(T_1+\varepsilon_{\rho -1}-x_n-1)}{W} \leq \frac{\rho l+x_n+1}{W} \leq t-\frac{(T_0+\varepsilon_0-x_n-1)}{W} \right\}.$$ 
If we now impose the conditions $\varepsilon_0>\rho-T_0$ 
and
$0\le x_0<x_1<\cdots<x_{\rho-1}\le \rho-1,$
then
$$T_0+\varepsilon_0-x_n-1>0.$$
This shows that the local averages used in the definition of the operator are taken only from the intervals located before the point $t$. Therefore, $K^{\psi, X}_W$ depends only on past local average samples and acts as a prediction operator. Moreover, the set $\Omega^n_t$ contains at most
$$
1+\left\lfloor
\tfrac{T_1-T_0+\varepsilon_{\rho-1}-\varepsilon_0}{\rho}
\right\rfloor
$$
integer values of $l$. Hence, the evaluation of $K_W^{\psi,X}f(t)$ requires at most
$$
\rho\left(
1+\left\lfloor
\tfrac{T_1-T_0+\varepsilon_{\rho-1}-\varepsilon_0}{\rho}
\right\rfloor
\right)
$$
past local average samples of $f$. 
Since $(\psi,X)$ satisfies the BSF conditions of order $r$, we  have the following result.

\begin{theorem}\label{prediction}
Let $\varphi\in\mathcal{F}_{r,\epsilon}^\rho$ and let
$
X=\{x_n+\rho l:\;0\le n\le \rho-1,\; l\in\mathbb Z\}$
be a PNS set with
$0\le x_0<x_1<\cdots<x_{\rho-1}\le \rho-1.$ Let the functions $\varphi_n$ be compactly supported with support $[T_0,T_1].$
For $n=0,\dots,\rho-1$, define
$$
\psi_n(t)=\sum_{p=0}^{\rho-1}a_p\,\varphi_n(t-\varepsilon_p), \quad 
a_p = \prod_{\substack{q=0 \\ q\neq p}}^{\rho-1} \frac{\varepsilon_q}{\varepsilon_q - \varepsilon_p},$$ where $\rho-T_0<\varepsilon_0 < \varepsilon_1 < \cdots < \varepsilon_{\rho-1}.$ If $(\varphi,X)$ satisfies the BSF conditions of order $r$, then the sampling Kantorovich operator $K^{\psi,X}_W$ satisfies 
$$
\left\|K^{\psi,X}_Wf - f\right\|_p
\leq
C \,W^{-r} \, \left\|f^{(r)}\right\|_p, \text{ for all }f \in \mathcal{W}^{r}_p(\mathbb{R}),
$$
where $C>0$ is a constant independent of $f$ and $W$. In addition, $K_W^{\psi,X}f(t)$ can be computed using at most 
$$
\rho\left(
1+\left\lfloor
\tfrac{T_1-T_0+\varepsilon_{\rho-1}-\varepsilon_0}{\rho}
\right\rfloor
\right)
$$
past local average samples of $f$.
\end{theorem}

\subsection{Numerical Implementation and Simulation}
We present numerical examples illustrating the performance of sampling Kantorovich operators. First, we construct a pair $(\varphi, X)$ satisfying the generalized BSF conditions using suitable linear combinations of shifted Gaussian functions. The corresponding approximation results are then verified for various values of the scaling parameter $W$.
Next, we construct $(\varphi, X)$ satisfying the BSF conditions using appropriate linear combinations of shifted B-splines. These B-spline-based generators are subsequently employed to construct prediction operators, and their effectiveness in reconstructing a smooth signal from a finite number of past local average samples is demonstrated. The MATLAB code used to generate the numerical examples and figures is available at the following link: \href{https://github.com/sreyasoman201/Sreya_Antony_BSF}{MATLAB code for the numerical illustrations}.\\
\noindent
\textbf{\underline{Example 1}:}\label{example1}
Consider the Gaussian function $g:\mathbb{R}\to\mathbb{R}$ defined by
$$g(t)=e^{-\pi t^2}.$$
The Fourier transform of $g$ is given by $\widehat{g}(\omega)=e^{-\pi \omega^2}.$
Taking $r=\rho=2$ in \eqref{phin}, the generator $\varphi= (\varphi_0, \varphi_1)$ is  given by
\begin{eqnarray*}
\varphi_0(t)=b_0 e^{-\pi t^2}+b_1 e^{-\pi(t-1)^2} \text{ and }
\varphi_1(t)=b_0 e^{-\pi(t-1)^2}+b_1  e^{-\pi(t-2)^2}.
\end{eqnarray*}
We now select the PNS set
$
X=2\mathbb{Z}\cup\left(\frac{1}{2}+2\mathbb{Z}\right).
$
For this choice of $X$, \eqref{system2} yields
$$
b_0=\frac{3}{2} \text{ and } b_1=-\frac{1}{2}.
$$
Hence, the pair $(\varphi, X)$ satisfies the generalized BSF conditions of order 2. Observe that the generator $\varphi \in \mathcal{F}_{2, \epsilon}^2$ for every $\epsilon>0$. Therefore,
by Corollary \ref{Loperator}, for any $0\leq \alpha<1$, we have
$$
\lim_{W\to\infty}\left\|L_{W, \alpha}^{\varphi,X}f-f\right\|_p=0, \text{ for every } f\in \mathcal{W}_p^2(\mathbb{R}).
$$                                          We now approximate the real-valued function $f:\mathbb{R}\to\mathbb{R}$ defined by 
\begin{equation}\label{test_function}
f(t)=e^{-t^{2}/4}\left(\sin(2\pi t)+\cos(\pi t)\right),
\end{equation} with $\alpha=0.5$
using the operator $L_{W, 0.5}^{\varphi.X}.$ 
Figure \ref{figure3} illustrates the performance of the operator $L^{\varphi, X}_{W, 0.5}$ in approximating the signal $f$.
The results show that the obtained approximation closely follows the original function. Furthermore, as the sampling parameter $W$ increases, the approximation becomes more accurate, indicating the effectiveness of the operator in reconstructing the given signal.

 \begin{figure}[H]
    \centering
    \begin{subfigure}[b]{0.40\textwidth}
        \centering
        \includegraphics[height=3.0cm,width=\textwidth]{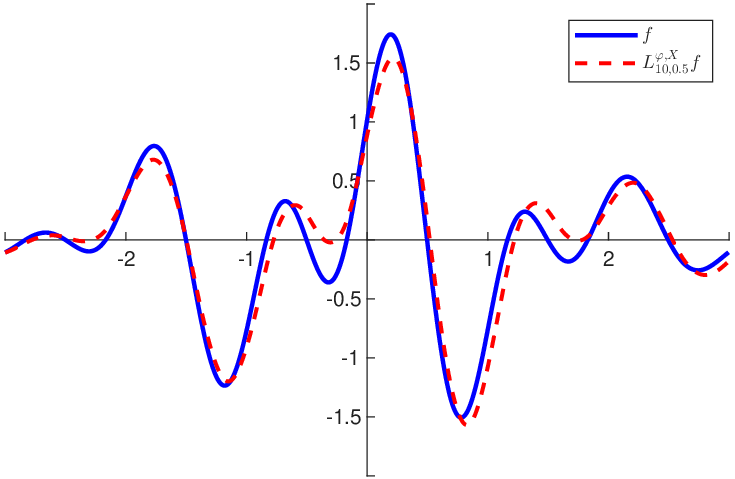}
    \end{subfigure}
    \hspace{0.5cm}
\begin{subfigure}[b]{0.40\textwidth}
        \centering
        \includegraphics[height=3.0cm, width=\textwidth]{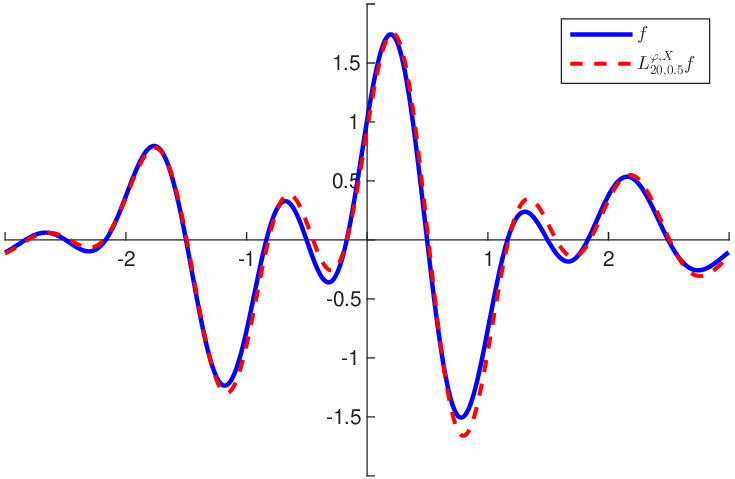}
    \end{subfigure}
    \begin{subfigure}[b]{0.40\textwidth}
        \centering
\includegraphics[height=3.0cm,width=\textwidth]{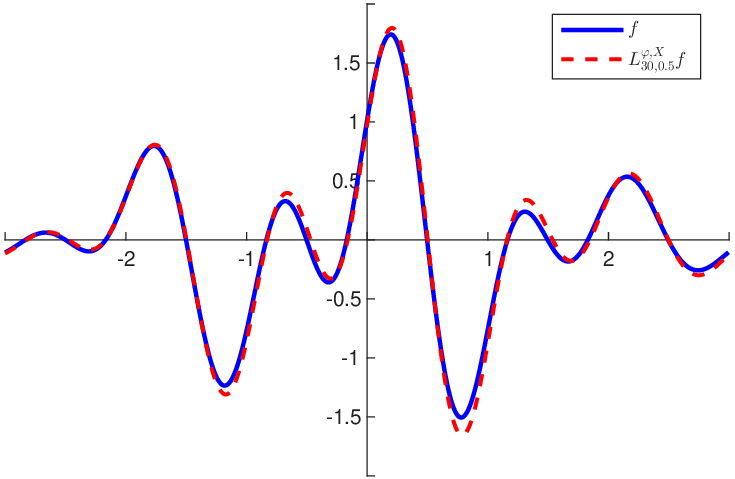}
    \end{subfigure}
 \caption{Approximations of $f$ using $L^{\varphi,X}_{W, 0.5}f$}
    \label{figure3}
\end{figure}
\noindent
\textbf{\underline{Example 2}:}\label{example2}
Recall that the B-splines are defined inductively as follows:
\begin{equation*}\label{bspline}
Q_1(t)=\chi_{[0,1)}(t), \quad
Q_{m+1}(t)=\displaystyle{\int\limits_{-\infty}^{\infty}Q_m(t-y)Q_1(y)~dy},~m\geq1.
\end{equation*}
The B-spline $Q_m$ is compactly supported on the interval $[0,m]$ and is $m-2$ times differentiable on $\mathbb{R}$. Also, $Q_m$ belongs to the space $\mathcal{F}_{4,\epsilon}^4$.

Let us consider the generator
$\theta=(\theta_0, \theta_1, \theta_2, \theta_3)$
defined by
\begin{align*}
\theta_{0}(t) &= -19 Q_4(t) + 19 Q_4(t+3) - \frac{13}{3} Q_4(t+2) + \frac{13}{3} Q_4(t+1), \\
\theta_{1}(t) &= \frac{208}{3} Q_4(t) - \frac{116}{3} Q_4(t+3) + \frac{40}{3} Q_4(t+2) - \frac{44}{3} Q_4(t+1), \\
\theta_{2}(t) &= -\frac{260}{3} Q_4(t) + \frac{82}{3} Q_4(t+3) - \frac{32}{3} Q_4(t+2) + \frac{46}{3} Q_4(t+1), \\
\theta_{3}(t) &= \frac{112}{3} Q_4(t) - \frac{20}{3} Q_4(t+3) + \frac{8}{3} Q_4(t+2) - 4 Q_4(t+1),
\end{align*}
and  the PNS set
$X=4\mathbb{Z}\cup \left(\frac{1}{4}+4\mathbb{Z}\right)\cup\left(\frac{1}{2}+4\mathbb{Z}\right)\cup \left(\frac{3}{4}+4\mathbb{Z}\right).$
It has been shown in \cite{SreyaGhoshSelvan} that $(\theta,X)$ satisfies \eqref{vanishingmoment1} for $r=4 $.
Solving the system \eqref{b_k}, we obtain
$$
(b_0, b_1, b_2, b_3) = \left(\dfrac{1}{4},\,\frac{13}{12},\,-\frac5{12},\,\frac1{12}\right).
$$
Since  the support of $Q_4$ is $[0,4]$, the functions $\theta_n$ are compactly supported in $[-3,4]$, and hence the functions $\varphi_n$ defined by \eqref{phi_n} are compactly supported in  $[-3,7]$. 
The graphs of the functions $\varphi_n$ are shown in Figure \ref{fig_varphi}.

We now choose
$$
\varepsilon_p = 7 + 0.25p, \qquad p = 1,2,3,4.
$$
Using \eqref{a_p2}, we obtain
$
(a_0, a_1, a_2, a_3) = (4960,\,-14384,\,13920,\,-4495).
$
So, the  new functions 
\[
\psi_n(t)
=
\sum_{p=0}^{3}
a_p\,\varphi_n(t-\varepsilon_p),
\]
are compactly supported in
$[4.25,15].$
Hence, by Theorem \ref{prediction}, it follows that the operator
$
K_{W}^{\psi,X}$ satisfies
$$\left\|K^{\psi,X}_Wf - f\right\|_p
\leq
C \,W^{-4} \, \left\|f^{(4)}\right\|_p, \text{ for every }f \in \mathcal{W}^{4}_p(\mathbb{R}).
$$  Moreover, the evaluation of $K_{W}^{\psi,X}f$   at any given time $t$ requires at most  12 past local average samples of $f$.

We approximate the function $f$ given in \eqref{test_function} using the prediction operator $K^{\psi, X}_{W}$. 
Figure \ref{figure1} illustrates how the proposed prediction operator  $K^{\psi, X}_{W}$ recovers the signal $f$ from the finite number of past local average samples. The results show that the predicted signal closely matches the original function, and increasing the parameter $W$ further improves the prediction accuracy.
\begin{figure}[ht]
    \centering
    \includegraphics[width=0.7\textwidth]{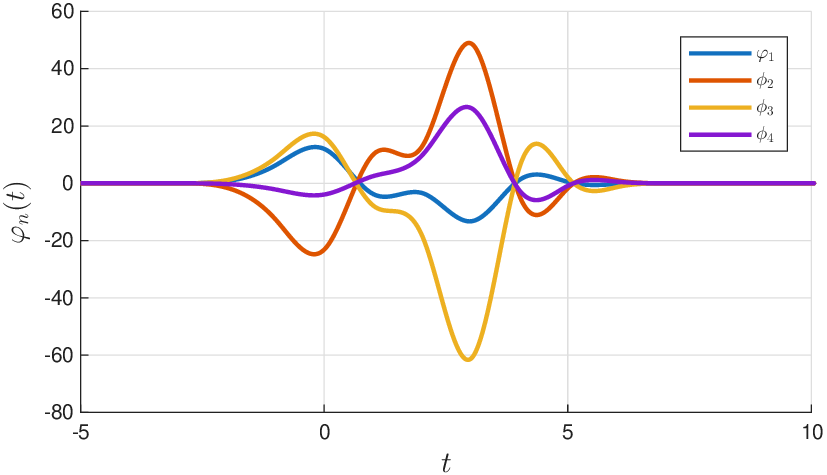}
    \caption{The functions $\varphi_n, ~n=0,1,2,3.$}
    \label{fig_varphi}
\end{figure}

\begin{figure}[H]
    \centering
    \begin{subfigure}[b]{0.40\textwidth}
        \centering
        \includegraphics[height=3.0cm,width=\textwidth]{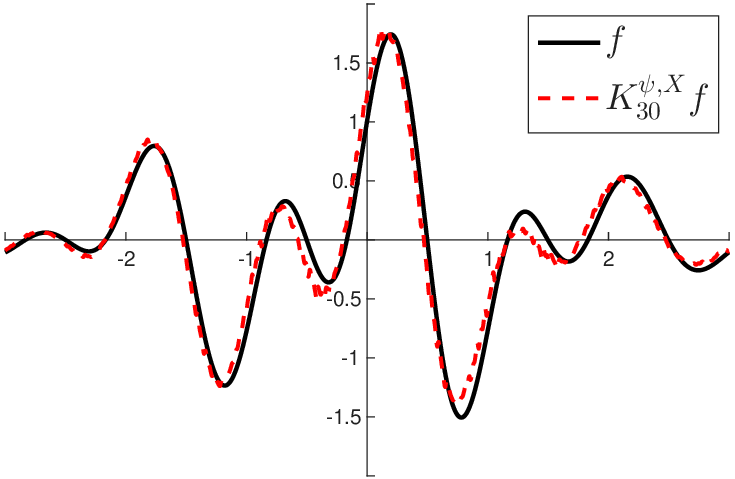}
    \end{subfigure}
    \hspace{0.5cm}
\begin{subfigure}[b]{0.40\textwidth}
        \centering
        \includegraphics[height=3.0cm, width=\textwidth]{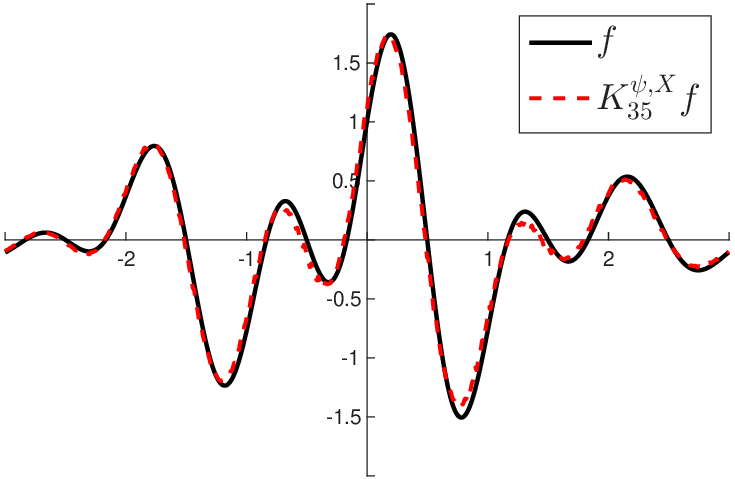}
    \end{subfigure}
    \begin{subfigure}[b]{0.40\textwidth}
        \centering
\includegraphics[height=3.0cm,width=\textwidth]{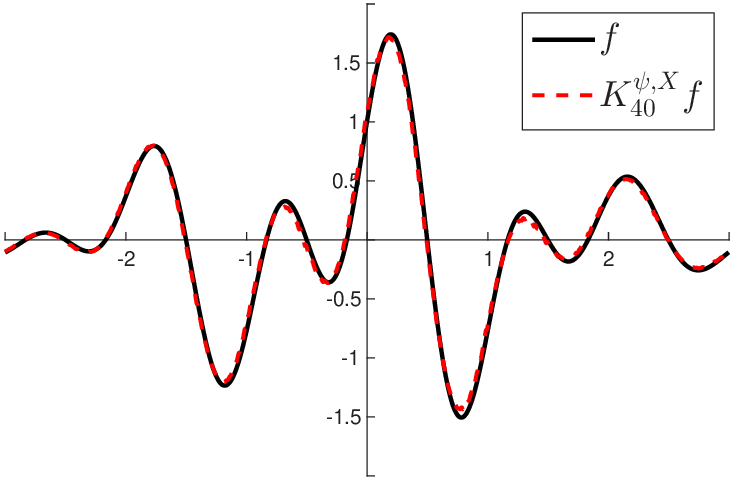}
    \end{subfigure}

    \caption{Approximations of $f$ using $K^{\psi,X}_{W}f$}
    \label{figure1}
\end{figure}

\section{Conclusion}
We introduced the Bernoulli–Strang–Fix conditions and their generalized versions for vector-valued generators and periodic nonuniform sampling sets. We then established exact and asymptotic polynomial reproduction properties of the associated sampling Kantorovich operators and analyzed their approximation and convergence behavior. We also demonstrated that these operators effectively predicted signals from a finite number of past local average samples. Finally, numerical experiments with Gaussian functions and B-splines validated the theoretical approximation and prediction results. The developed theory may serve as a foundation for future applications in signal reconstruction and prediction from averaged samples, including problems in signal processing, imaging, communications, and sensor networks. Future research will extend this framework to multidimensional settings and more general sampling schemes.

\section*{Acknowledgment}
The author (ST) gratefully acknowledges the Ministry of Education, Government of India, for supporting this research through the Prime Minister’s Research Fellowship and Grant (PMRF ID: 1603259). The authors acknowledge the use of ChatGPT (OpenAI) for language editing and improving the clarity of the manuscript.

\textbf{Data Availability.}
Data sharing does not apply to this article as no datasets were generated or analyzed during the current study. 

\textbf{Conflict of interest.} The authors declare that there is no conflict of interest.

\end{document}